\documentclass[aps,twocolumn,superscriptaddress,prl,10pt]{revtex4-1}
\usepackage{amsmath,amsfonts,amssymb,amsthm,graphicx,enumerate,bbm}
\usepackage[caption=false]{subfig}
\usepackage{mathtools}
\usepackage[colorlinks]{hyperref}
\usepackage[normalem]{ulem}

\newcommand{\ket}[1]{\left \vert #1 \right \rangle}

\newcommand{\VV}{\mathcal{V}}
\newcommand{\GG}{\mathcal{G}}
\newcommand{\hh}{\mathcal{H}}

\newcommand{\id}{\mathbbm{1}}
\newcommand{\Id}{\id}

\newcommand{\rr}{\mathbbm{R}}

\newcommand{\dd}{\mathrm{d}}
\DeclareMathOperator{\supp}{supp}
\DeclareMathOperator{\polylog}{polylog}
\DeclareMathOperator{\poly}{poly}

\newcommand{\expval}[3]{\left \langle #1 \vphantom{#2} \vphantom{#3} \right \vert #2 \left \vert \vphantom{#1}\vphantom{#2} #3 \right \rangle}
\newcommand{\braket}[2]{\left \langle #1 \vphantom{#2} \right \vert \left. #2 \vphantom{#1} \right \rangle}
\newcommand{\ve}{\varepsilon}
\newcommand{\be}{\begin{equation}}
\newcommand{\ee}{\end{equation}}
\newcommand{\Mu}{\Upsilon}

\newtheorem{theorem}{Theorem}
\newtheorem{lemma}[theorem]{Lemma}

\theoremstyle{definition}

\begin{document}

\title{Rapid adiabatic preparation of injective PEPS and Gibbs states}

\author{Yimin Ge}
\author{Andr\'as Moln\'ar}
\author{J. Ignacio Cirac}
\affiliation{Max-Planck-Institut  f{\"u}r Quantenoptik, D-85748 Garching, Germany}

\begin{abstract}
We propose a quantum algorithm for many-body state preparation. It is especially suited for injective PEPS and thermal states of local commuting Hamiltonians on a lattice. We show that for a uniform gap and sufficiently smooth paths, an adiabatic runtime and circuit depth of $O(\polylog N)$ can be achieved for $O(N)$ spins. This is an almost exponential improvement over previous bounds. The total number of elementary gates scales as $O(N\polylog N)$. This is also faster than the best known upper bound of $O(N^2)$ on the mixing times of Monte Carlo Markov chain algorithms for sampling classical systems in thermal equilibrium.
\end{abstract}
\maketitle

Quantum computers are expected to have a  deep impact in the simulation of {\it quantum} many-body systems, as initially envisioned by Feynman \cite{Feynman1982}. In fact, quantum algorithms have potential applications in diverse branches of science, ranging from condensed matter physics, atom physics, high-energy physics, to  quantum chemistry \cite{[{See special issue on quantum simulation in }] [{}]NatPhysSim}.  Lloyd \cite{Lloyd1996} was the first to devise a quantum algorithm to simulate the dynamics generated by few-body interacting Hamiltonians. When combined with the adiabatic theorem \cite{Kato1950,Farhi2000}, the resulting algorithms allow one to prepare ground states of local Hamiltonians, and thus to investigate certain quantum many-body systems at zero temperature. Quantum algorithms have also been introduced to prepare 
so-called projected entangled pair states (PEPS) \cite{Verstraete2004,Schwarz2012,PhysRevA.88.032321}, which are believed to approximate ground states of local gapped Hamiltonians. Furthermore, quantum algorithms have also been proposed to sample from Gibbs distributions  \cite{Terhal2000,PhysRevLett.105.170405,Temme2011,Yung2012,PhysRevLett.108.080402,Kastoryano2014}, which describe physical systems in thermal equilibrium. The computational time of most of these algorithms is hard to compare with that of their classical counterparts, as it depends on specific (e.g., spectral) properties of the Hamiltonians which are not known beforehand. However, they do not suffer from the sign problem \cite{SuzukiQMCM}, 
which indicates that they could provide significant speedups.

Quantum computers may also offer advantages in the simulation of {\it classical} many-body systems. For instance, quantum annealing algorithms \cite{PhysRevE.58.5355,Farhi2001} have been devised to prepare the lowest energy spin configuration of a few-body interacting classical Hamiltonian, which has obvious applications in optimization problems. 
Quantum algorithms have also been proposed to 
sample from their  
Gibbs distributions at finite temperature \cite{PhysRevA.76.042306,PhysRevE.56.3661,Somma2007,Wocjan2008,Somma2008,PhysRevA.82.060302}. Apart from 
applications
in classical statistical mechanics, 
similar problems appear in other areas of intensive research, e.g., {machine}
 learning. 
 Speedups as a function of spectral gaps have been analysed in Refs.~\cite{Somma2008,Wocjan2008,Yung2012}; the scaling with large system sizes, which is of particular interest for applications in deep machine learning \cite{[{See, for example, }]Bengio2009}, is however not optimal.

In this Letter we propose and analyse a quantum algorithm to {\em efficiently} prepare a particular set of states. This set contains two classes relevant for lattice problems: (i) injective PEPS \cite{Perez-Garcia2007}; (ii) Gibbs states of locally commuting Hamiltonians. 
Class (ii) contains 
all classical Hamiltonians, and thus the quantum algorithm allows us to 
sample Gibbs distributions of classical problems at finite temperature.

Our algorithm outperforms all other currently known algorithms for these two problems in the case 
that the minimum gap $\Delta$ occurring in the adiabatic paths (to be defined below) is lower bounded by a constant. 
We show that the computational time for a quantum computer, given by the number of elementary gates in a quantum circuit, scales only as
 \be \label{eq:gates}
 T = O\left({N \polylog \left( N/\epsilon\right) }\right),
\ee
where $N$ is the number of local Hamiltonian terms, $\ve$ the allowed error in trace distance and the degree of the polynomial depends on the geometry of the lattice. 
Note that an obvious lower bound on the computational time is $\Omega(N)$, as each of the spins has to be addressed at least once. Thus, Eq.~\eqref{eq:gates} is almost optimal. 
Furthermore, the algorithm is parallelisable, so that the depth of the circuit becomes
\be \label{eq:depth}
	D = O\left( \polylog (N/\ve) \right).
\ee
This parallelisation may also become very natural and relevant in analog quantum simulation, as is the case for atoms in optical lattices \cite{RevModPhys.80.885}.

One of the best classical algorithms to sample according to the Gibbs distribution of a general classical Hamiltonian is the well-known Metropolis algorithm \cite{Metropolis1953}. The currently best upper bound to its computational time is $T=O(N^2/\Delta_{\rm stoch})$ \cite{LevinPeresWilmer2008}, where $\Delta_{\rm stoch}$ is the gap of the generator of the stochastic matrix. We will see that given any stochastic matrix, one can always construct a quantum adiabatic algorithm with the same gap $\Delta=\Delta_{\rm stoch}$,  and thus we obtain a potential quantum speedup of almost a factor of $N$. Under parallelisation, the circuit depth is almost exponentially shorter. Our algorithm  to prepare injective PEPS also provides a better scaling than the one presented in Ref.~\cite{Schwarz2012}.

\begin{figure}[h]
\includegraphics[width=0.45\textwidth]{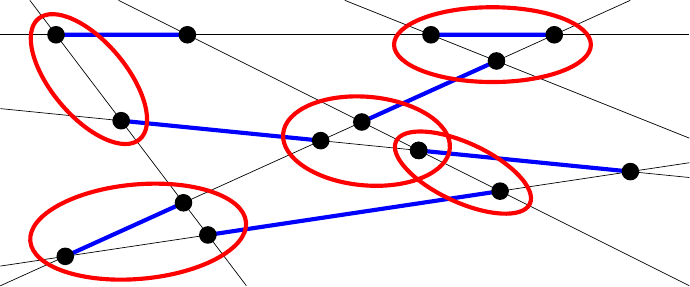}
\caption{The general class of states, Eq.~\eqref{PsiN}. Finite range operators (red) acting on a collection of maximally entangled pair states (blue) distributed on a graph.}\label{fig:graph}
\end{figure}

The class of states we consider in this Letter can be thought of as commuting finite range operators acting on a set of maximally entangled pair states (Fig.~\ref{fig:graph}). More precisely, consider a regular lattice in some dimension, and let $\GG=({\cal V},{\cal E})$ be the associated (infinite) graph. 
 We endow $\GG$ with a distance $\dd$,  the minimum number of edges separating two vertices in $\VV$. We associate a $d$-dimensional Hilbert space, $\hh_v$, to each of the vertices $v\in\VV$. 
Consider the set ${\bf\Lambda}$ 
of interaction supports, 
i.e., ${\bf\Lambda}$ is a collection of sets of vertices whose relative distance is at most a constant $R$, the interaction length, and consider for each $\lambda\in{\bf\Lambda}$ an interaction $Q_\lambda$  which is an operator supported on $\bigotimes_{v\in\lambda}\hh_v$. We assume that they are strictly positive, $\Id\ge Q_\lambda>q_0\Id$, and mutually commute, $[Q_\lambda,Q_{\lambda'}]=0$. Consider also a set ${\bf\Mu}$ of mutually excluding pairs of neighbouring vertices. Moreover, let $\Lambda_N$ be a finite subset of ${\bf\Lambda}$ with $|\Lambda_N|=N$, and define 
 \be
 \label{PsiN}
 |\phi_N\rangle \propto \prod_{\lambda\in \Lambda_N} Q_{\lambda} \bigotimes_{\mu\in \Mu_N} |\phi^+\rangle_\mu,
 \ee
where 
$\Mu_N=\{\mu\in{\bf\Mu}\mid\mu\cap(\bigcup_{\lambda\in\Lambda_N}\lambda)\neq\emptyset\}$ is the set of pairs with a vertex in $\Lambda_N$, and
$|\phi^+\rangle=\sum_{i=1}^d\ket{ii}$ is an unnormalized maximally entangled state between the pairs of vertices in $\Mu_N$. We will give a quantum algorithm to prepare the state Eq.~\eqref{PsiN}, and analyse the runtime as a function of $N$ and other spectral properties. In the following, we drop the subindex $N$ to ease the notation.

\begin{figure}[h]
	\subfloat[][]{
		\includegraphics[width=.23\textwidth]{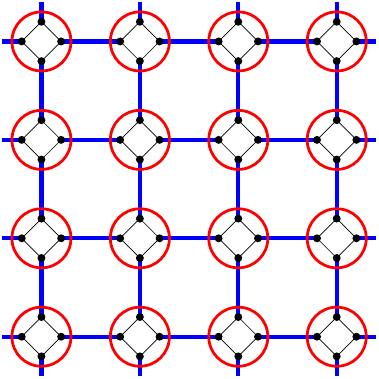}
		\label{subfig:peps}
	}
	\subfloat[][]{
		\includegraphics[width=.23\textwidth]{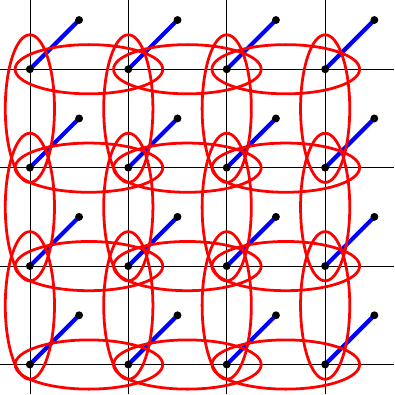}
		\label{subfig:thermal}
	}
	\caption{(a) Projected entangled pair states. (b) Purification of a thermal state. For each system qudit, we introduce an ancilla to be placed in a maximally entangled pair with its system particle, then apply $e^{-\beta H/2}$ to the system.}
	\label{fig:states}
\end{figure}

As mentioned above, Eq.~\eqref{PsiN} includes two relevant classes of states. The first is the class of injective PEPS. The graph is composed of nodes, each of them including a set of vertices (Fig.~\ref{subfig:peps}). In this case, $\Mu$ contains pairs of vertices in nearest neighbor nodes, whereas $\Lambda$ contains each node. The operators $Q_\lambda$ act on different nodes, and therefore trivially commute. The resulting state is just a PEPS, which is injective since each $Q_\lambda$ is invertible. In fact, every injective PEPS can be expressed in this form up to a local unitary using a QR decomposition. The second class is the class of Gibbs states of commuting Hamiltonians \cite{[{See also }][{ for a similar parent Hamiltonian construction.}] Feiguin13}. To see this, consider the graph which contains sites composed of two vertices, one of them is called ``system'' and the other ``ancilla''. The set $\Mu$ contains all sites, whereas $\Lambda$ contains interacting system vertices (Fig.~\ref{subfig:thermal}). The relation with Gibbs states is evident if we write $Q_\lambda=e^{-\beta h_\lambda/2}$, where $||h_\lambda||<1$, and take into account that they mutually commute. It is easy to see that if we trace the ancillas, we obtain
 \be
 \rho \propto e^{-\beta H},
 \ee
where $H=\sum_{\lambda\in\Lambda}h_\lambda$.

The state Eq.~\eqref{PsiN} is the unique ground state of a frustration-free local Hamiltonian that can be written as 
\begin{equation}\label{eq:globalG}
G = \sum_{\mu\in \Mu} G_\mu, 
\end{equation}
with
 \be \label{eq:Gmu}
 G_\mu = \left(\prod_{\lambda\in \Lambda_\mu} Q_\lambda^{-1} \right)P_\mu \left(\prod_{\lambda\in \Lambda_\mu} Q_\lambda^{-1}\right),
 \ee
where $\Lambda_\mu =\{\lambda\in\Lambda\mid\lambda\cap\mu\neq\emptyset\}$ is the set of supports whose interactions act nontrivially on $\mu$,
and $P_\mu$ is the projector onto the subspace orthogonal to $|\phi^+\rangle_\mu$. Notice that since each $G_\mu$ is supported in a region of radius $R$ around $\mu$, $G$ is indeed local.
 
The state Eq.~\eqref{PsiN} can be prepared using an adiabatic algorithm. 
For that, we define a path $Q_\lambda(s)$ with unique ground state $\ket{\phi(s)}$,
where $s\in[0,1]$, with $Q_\lambda(0)=\Id$ and $Q_\lambda(1)=Q_\lambda$. We can choose $Q_\lambda(s)= (1-s)\Id+sQ_\lambda$. In the case of the thermal state, we can also choose $Q_\lambda(s)=e^{-\beta sh_\lambda/2}$. 
Then, by starting with $\ket{\phi(0)}$ and changing the parameter $s:0\to 1$ sufficiently slowly, we will end up in the desired state $\ket{\phi(1)}$. 
The runtime for this preparation, as measured by the number of elementary quantum gates, is unpractical, however, as it scales as $T=O(N^4\Delta^{-3}\varepsilon^{-1}\polylog(N/\ve))$, where $\ve$ is the tolerated error and $\Delta$ is the minimum spectral gap along the path. Indeed, the adiabatic  theorem \cite{Jansen2007} gives an adiabatic runtime of $\tau=O(N^2\Delta^{-3}\varepsilon^{-1})$ so that Hamiltonian simulation \cite{Berry2014} gives $T=O(\tau N^2\polylog(N/\ve))$. 

To obtain a better scaling, we first use a variant of the adiabatic theorem with almost exponentially better runtime dependence on the error using a sufficiently smooth reparameterisation of the Hamiltonian path. The quadratic scaling of the runtime with the derivative of the Hamiltonian, however, leads to an   
unpractical dependence on $N$ since the Hamiltonian contains $O(N)$ terms that change with time. To avoid this, we change the $Q$'s individually, leading to an adiabatic runtime of {$\tau=O(N\log^{1+\alpha}\left(N/\ve\Delta\right)\Delta^{-3})$}. This, however, uses Hamiltonians acting on the whole system, despite only the change of a single $Q$, which would result in an additional factor of $O(N^2\polylog(N/\ve))$ for the computational time measured by the number of elementary gates. We circumvent this problem by using Lieb-Robinson bounds \cite{lieb1972} and the frustration freeness to show that {under the assumption of a uniformly lower bounded spectral gap}, it is at each step sufficient to evolve with a Hamiltonian acting only on $O(\polylog(N/\ve))$ sites instead of the full lattice.

Thus, define a sequence of 
$N$ 
Hamiltonian paths by enumerating the elements of $\Lambda$ as $\lambda_1,\ldots,\lambda_N$, and define
\begin{equation}\label{eq:localisedG}
	G_n(s) = \sum_{\substack{\mu\in\Mu\\ \dd(\mu, \lambda_n) < \chi\log^{1+\alpha}( N/\ve)}}G_{n,\mu}(s)
\end{equation}
for $n=1,\ldots,N$, where 
$\chi$ is a constant control parameter, and
\begin{equation}
 G_{n,\mu}(s)
  = \left(\prod_{\substack{m \\ \lambda_m\in \Lambda_\mu}}  A_{n,m}^{-1}(s) \right)P_\mu \left(\prod_{\substack{m \\ \lambda_m\in \Lambda_\mu}}   A_{n,m}^{-1}(s)\right)
\end{equation}
with 
\begin{equation}\label{eq:Anlambda}
A_{n,{m}}(s) = \begin{cases}
	Q_{\lambda_m}(1) &  m < n \\
	Q_{\lambda_m}(s) & m=n \\
	Q_{\lambda_m}(0)=\Id &  m>n.
\end{cases}
\end{equation}
Notice that $G_n$ is supported on a region of radius $O(\log^{1+\alpha}(N/\ve))$ and $\dd G_n /\dd s$ is supported on a region of bounded size. 
By reparameterising $G_n(s)\to G_n(f(s))$ with $f$ a function in the Gevrey class $1+\alpha$, we can assume $G_n$ to be in the same Gevrey class %\cite{Note1}. 
\footnote{Recall that a function $f:[0,1]\rightarrow\rr^m$ is in the Gevrey class $1+\alpha$ \cite{Gevrey1918} (with respect to the norm \unexpanded{$\|.\|$}) if there exist constants $c,K>0$ such that \unexpanded{$\|{\dd^k f(s)}/{\dd s^k}\|\leq Kc^k(k!)^{1+\alpha}$} for all $k$. It is well known \cite{Ramis1978} that $f(s)=\int_0^sf_\alpha(t)\dd t/\int_0^1f_\alpha(t)\dd t$, with $f_\alpha(t)=\exp(-1/((1-t)t)^{1/\alpha})$, is in the Gevrey class $1+\alpha$ for all $\alpha>0$.}.

Consider the sequence of Schr\"odinger equations
\begin{equation}
	i\frac{\dd}{\dd t}\ket{\psi_n} = G_n\left(\frac t{\tau_n}\right)\ket{\psi_n},\quad \ket{\psi_{n+1}(0)}=\ket{\psi_n(\tau_n)}
\end{equation}
for times $\tau_n=O(\log(N/\ve)^{1+\alpha})$, stating in $\ket{\psi_1(0)}=\ket{\phi(0)}$, the trivial ground state of $G_1(0)$.  
The algorithm proceeds by running Hamiltonian simulation \cite{Berry2014} on this sequence of adiabatic evolutions. Since at all times we only evolve with Hamiltonians acting on $O(\polylog(N/\ve))$ sites, the number of gates only grows as Eq.~\eqref{eq:gates}. 
Moreover, the evolution of consecutive $G_n$s can be parallelised if their support is disjoint, i.e., if $G_n,\ldots,  G_{n+l}$ have disjoint supports, the subsequence can be replaced by their sum without altering the evolution. Since $|\supp G_n|=O(\polylog (N/\ve))$, it is clear that an ordering of the $\lambda_n$ can be chosen such that subsequences of length $\Omega(N/\polylog(N/\ve))$ of the $G_n$s can be parallelised at a time, resulting in a circuit of depth Eq.~\eqref{eq:depth}, an almost exponential improvement over previous bounds.

In the following, we show that for a uniformly lower bounded gap, the error of the above algorithm is bounded by $\ve$. 
First, we use that under sufficient smoothness conditions on a Hamiltonian path $G(s)$, 
the  final error can be almost exponentially small in the adiabatic runtime. 
Indeed, as proven in the Supplemental Material, 
if $G$ is in the Gevrey class $1+\alpha$ and $\dd^k G/\dd s^k=0$ at $s=0,1$ for all $k\geq1$, then an adiabatic runtime of
\begin{equation}\label{eq:expErrorT}
 \tau=O\left(\log^{1+\alpha}\left(\frac{K}{\ve\Delta}\right)\frac{K^2}{\Delta^3}\right)
\end{equation}
is sufficient for an error $\ve$, 
where $\Delta$ is the minimum gap of $G(s)$ and $K=|\supp \dd G/\dd s|$ if $G$ is local. 
The required smoothness conditions can always be achieved with a suitable reparametrisation of the path $G(s)\rightarrow G(f(s))$.

This allows us to implement the global change of the Hamiltonian, Eq.~\eqref{eq:globalG}, as a sequence of $N$ local changes. Define the sequence of Hamiltonian paths,
\begin{equation}\label{eq:entireG}
	\tilde G_n(s) = \sum_{\mu\in\Mu} G_{n,\mu}(s).
\end{equation}
Notice that Eq.~\eqref{eq:entireG} is the same as Eq.~\eqref{eq:localisedG}, but contains all local terms $G_{n,\mu}$. 
The weak dependence on $\ve^{-1}$ in Eq.~\eqref{eq:expErrorT} ensures that the accumulated error under the sequential evolution with Eq.~\eqref{eq:entireG} remains small. 
Indeed, for a final error $\ve$, it is sufficient that the evolution with each $\tilde G_n$ in this sequence only generates an error of at most $\ve/N$. Equation~\eqref{eq:expErrorT} and 
$|\supp{\dd\tilde G_n}/{\dd s}|=O(1)$ 
imply that this can be achieved in a time {$\tau_n=O\left(\log^{1+\alpha}(N/\ve\Delta_n)\Delta_n^{-3}\right)$}, where $\Delta_n$ is the minimum spectral gap of $\tilde G_n$. 
A decomposition into a circuit then requires 
$T=O\left(N^3\polylog(N/\ve\Delta)\Delta^{-3}\right)$ 
 elementary gates, where $\Delta=\min_n\Delta_n$. This is already an improvement by a factor $N$ over the naive change of the entire Hamiltonian, assuming similar behaviour of $\Delta$ compared to the spectral gap of the original path $G(s)$.

Assuming that $\Delta=\Omega(1)$, we can further improve on this using Lieb-Robinson bounds to localise the effect of the adiabatic change. Indeed, we show in the Supplemental Material that local terms in Eq.~\eqref{eq:entireG} which are supported at a distance $\Omega(\log^{1+\alpha} (N/\ve))$ away from the support of ${\dd}\tilde G_n/{\dd s}$ do not significantly contribute to the unitary evolution generated by Eq.~\eqref{eq:entireG}. This allows the replacement of Eq.~\eqref{eq:entireG} with Eq.~\eqref{eq:localisedG} without significantly altering the evolution and thus the final state. Notice that 
$G_n$ only acts on $O(\polylog(N/\ve))$ sites 
 and $\tau_n=O(\polylog(N/\ve))$ for all $n$. Thus, its unitary evolution can be simulated with only $O(\polylog(N/\ve))$ gates. Hence, we finally obtain a number of gates in the circuit model that grows only as Eq.~\eqref{eq:gates} for a constant error and lower bounded spectral gap. Using the described parallelisation, we finally obtain a circuit depth Eq.~\eqref{eq:depth}, as claimed.

In the analysis above, we have assumed a gap $\geq\Delta$ along all $N$ paths. This assumption can in fact be relaxed to a gap at either $s=0$ or $s=1$ (see Supplemental Material), using the positivity condition on $Q_\lambda$. 
We thus say that the system has a uniformly lower bounded gap $\Delta$ if for all finite subsets $\Lambda\subset{\bf\Lambda}$, the Hamiltonian Eq.~\eqref{eq:globalG} has a spectral gap $\geq\Delta$. Under this assumption % \cite{Note2}, 
\footnote{In fact, this assumption can be relaxed to only hold for the $N$ subsets $\Lambda_n=\{\lambda_1,\ldots,\lambda_n\}, n=1,\ldots,N$, for each problem size $N$, i.e., the sets being used in the algorithm, instead of all finite subsets.}, 
the circuit depth  Eq.~\eqref{eq:depth} can be guaranteed %\cite{Note3}. 
\footnote{For thermal states, it can be shown perturbatively that there exists some constant value $\beta_c$ such that the condition on the uniform spectral gap is always satisfied for $\beta<\beta_c$. For practical applications, it can however be hoped that this condition is satisfied for significantly higher $\beta$.}.

For the preparation of thermal states of classical Hamiltonians $H$, it is natural to compare these results with the performance of classical Monte Carlo Markov chain algorithms for Gibbs sampling such as the Metropolis algorithm or Glauber dynamics. Notice that due to the nature of their implementation, a fair comparison of performance should compare the mixing time of a discrete-time Markov chain to the number of elementary quantum gates, whereas the mixing time of a continuous-time Markov chain should be compared to the circuit depth. The best known upper bound  on the discrete mixing time for Monte Carlo Markov chain algorithms for sampling from Gibbs distributions of classical Hamiltonians given just the promise of a spectral gap scales as $O(N^2)$. Although under certain additional assumptions such as translational invariance \cite{MO1994,MO1994-2}, weak  mixing in two dimensions \cite{MOS1994}, or high temperature \cite{Hayes2006}, the existence of a logarithmic Sobolev constant and hence the rapid (discrete) mixing time of $O(N\log N)$ can be proven, no such proof exists for the general case to the best of our knowledge.
Our scheme thus outperforms classical Monte Carlo algorithms whenever rapid mixing cannot be shown even in the presence of a uniform gap. 

Note that any classical Monte Carlo algorithm can be realised as an adiabatic algorithm, as has, e.g., been observed in Ref.~\cite{Somma2007}. Indeed, if $S$ is the generator matrix of a continuous-time Monte Carlo algorithm that satisfies detailed balance with respect to the Gibbs distribution, $G=-e^{\beta H/2}Se^{-\beta H/2}$ is Hermitian. This Hamiltonian has the same spectrum as $-S$ and has the unique ground state $e^{-\beta H/2}\ket{+}^{\otimes N}$. For classical Hamiltonians $H$, this state has the same measurement statistics as $\rho_\beta$ for observables that are products of $\sigma_Z$. {By introducing an ancilla for every particle and applying the map $\ket{i}\mapsto\ket{ii}$, the purified version of the thermal state can also be recovered, and its parent Hamiltonian  has the same spectrum as $-S$ within the symmetric subspace. Thus, any classical system with a uniform spectral gap for the generator matrix can be turned into a rapid adiabatic algorithm.
}

For quantum Hamiltonians, notice the restriction to commuting local terms. For noncommuting local terms, an approximate quasilocal parent Hamiltonian can be considered above some constant temperature that allows the preparation in polynomial time. We describe this procedure in the Supplemental Material. 

For the preparation of injective PEPS, the given algorithm is similar to Ref.~\cite{Schwarz2012}, which, however, requires a runtime of $O(N^4)$ in the uniformly gapped case, due to the use of phase estimation and the ``Marriot-Watrous trick'',  which are computationally expensive for large systems. The better runtime of the present algorithm is largely due to replacing these subroutines by a local adiabatic change.

Throughout the analysis of this Letter, we focused on the case where a uniform constantly lower-bounded spectral gap is assumed. This assumption is only used to obtain a small number of elementary gates and circuit depth, whereas the adiabatic runtime of {$\tau=O(N\log^{1+\alpha}(N/\ve\Delta)\Delta^{-3})$} is independent of this assumption  \cite{[{See also }][{ for lower bounds on generic adiabatic runtimes.}] PhysRevA.81.032308}. In contrast, the runtime of the algorithm to prepare PEPS given in Ref.~\cite{Schwarz2012} only grows as $T\sim\Delta^{-1}$ for small gaps, and for thermal states, algorithms based on quantum walks, phase estimation, and the quantum Zeno effect have been proposed with a runtime of $T\sim\Delta^{-1/2}$ \cite{Somma2008,Wocjan2008,Yung2012}, albeit with worse scaling in the system size. We believe that similar techniques can be applied to our scheme of local changes to obtain a good scaling of the runtime for both large system sizes and small spectral gaps. {Moreover, it would be interesting to investigate if this scheme of local changes can also be applied to  speed up \emph{classical} Monte Carlo algorithms.}

We have also shown that the algorithm can be parallelised, thus giving rise to a circuit depth that scales only polylogarithmically with $N$. This is particularly attractive for certain experimental realisations of analog quantum simulators, such as with atoms in optical lattices \cite{Bloch2012} or trapped ions \cite{Blatt2012}, where this could be carried out in a very natural way.

\begin{acknowledgments}
We thank A.~Lucia, D.~ P\'erez-Garc\'ia, A.~Sinclair, and D.~Stilck~Fran\c{c}a for helpful discussions. This work was supported by the EU Integrated Project SIQS.
\end{acknowledgments}

\onecolumngrid
\appendix

\newpage

\section*{Supplemental Material}

\section{Proof of the adiabatic theorem with almost exponential error decay}

In this section, we prove a variant of the adiabatic theorem that only requires a runtime almost exponentially small in the allowed error. Our proof largely follows the proof given in \cite{Nenciu1993}, which is based on the method of adiabatic expansion \cite{Hagedorn2002}. The adiabatic expansion in \cite{Hagedorn2002} establishes an approximation of the  time-dependent Schr\"odinger evolution in terms of the instantaneous ground state and its derivatives, but on its own does not necessarily imply an adiabatic theorem because it assumes a special initial state. Our proof, like \cite{Nenciu1993}, resolves this problem by exploiting the Gevrey-class condition which allows to satisfy these initial conditions, and uses this expansion to establish a bound on the required runtime. However, unlike \cite{Nenciu1993}, which only proves the almost exponential dependence of the runtime with respect to accuracy, our proof also explicitly establishes the dependence on all other parameters such as the spectral gap and the bound on the Hamiltonian derivatives 
\footnote{Exponentially small errors have also been reported in \cite{Lidar2009}, however, \unexpanded{$\xi(n)$} appearing in Eq.~(22) of that paper should be defined as the supremum over \unexpanded{$S_\gamma$} instead of \unexpanded{$[0,1]$}, which implies a dependence of this quantity on \unexpanded{$N$}. Once this is taken into account, it is unclear how the arguments of that paper imply an exponentially small error for arbitrarily large runtimes. Nevertheless, numerical evidence in \cite{PhysRevA.82.052305} suggests that the error can be viewed as exponentially small for sufficiently small runtimes. }.

Consider a smooth path of Hamiltonians, $G(s)$, $s\in[0,1]$, acting on a finite-dimensional Hilbert space $\hh$. Let $\phi(s)$ \footnote{In the following, we will omit kets and bras to simplify notation} be the ground state of $G(s)$ and $\psi(\epsilon,s)$ the solution of the following Schr\" odinger equation:
\begin{equation}\label{eq:schroedinger}
  i \epsilon \dot{\psi}(\epsilon,s) = G(s)\psi(\epsilon,s), \quad 
  \psi(0)= \phi(0),
\end{equation}
where $1/\epsilon=\tau$ is the runtime of the adiabatic algorithm, and $\dot \ $ denotes derivative with respect to $s$. We assume furthermore that the ground state energy of $G(s)$ is 0 (i.e., we fix the phase of $\psi$) and that it has a gap at least  $\Delta$ throughout the whole path. By an appropriate choice of the phase of $\phi$, we can without loss of generality assume that $\braket{\dot \phi(s)}{\phi(s)}=0$. In the following, we will sometimes drop the explicit dependence on $s$ to simplify the notation. Unless otherwise stated, $\|.\|$ will always denote the operator norm for operators and the Euclidean vector norm for vectors (it will always be clear from the context which one is used). In this section, we prove the following theorem.

\begin{theorem}\label{thm:AT_Gevrey}
  Suppose that all derivatives of $G$ vanish at 0 and at 1, and moreover that it satisfies the following Gevrey condition: there exist non-negative constants $K$, $c$ and $\alpha$  such that for all $k\geq 1$,
  \begin{equation}
   \label{eq:2}
   \|G^{(k)}\|\leq K c^k \frac{[k!]^{1+\alpha}}{(k+1)^2}.
  \end{equation}
   Then, 
   \begin{equation}\label{eq:exp_decay_error}
     \min_\theta\|\psi(\epsilon,1)-e^{i\theta}\phi(1)\|\leq     8 ce\frac{K}{\Delta}\left(\frac{4\pi^2}{3}\right)^3\cdot \exp\left\{- \left(\frac{1}{4ec^2}\left(\frac{3}{4\pi^2}\right)^5\tau \frac{\Delta^3}{K^2}\right)^{\frac{1}{1+\alpha}} \right\}.
   \end{equation}
\end{theorem}

Notice that we don't require the Gevrey condition \eqref{eq:2} to hold for $k=0$. Therefore, in the application of Theorem~\ref{thm:AT_Gevrey} in the main text, $K=O(1)$ since along the paths $\tilde G_n(s)$ (as defined in \eqref{eq:entireG} in the main text), only $O(1)$ local terms change. 

\begin{proof}[Proof of Theorem~\ref{thm:AT_Gevrey}]
  Following the adiabatic expansion method from  \cite{Nenciu1993,Hagedorn2002}, we search $\psi(\epsilon,s)$ in the form of an asymptotic series expansion by constructing vectors $\phi_j(s)$, $s\in[0,1]$, $j\geq 0$, such that for all $M>0$, 
  \begin{equation}\label{eq:expansion}
    \left\|\psi(\epsilon,s) - \sum_{j=0}^{M-1} \phi_j(s) \epsilon^j \right \|= O(\epsilon^M). 
  \end{equation}
  We first show an explicit expression for $\phi_j$ provided that such an expansion exists. Second,we prove that the expansion really exists if $G^{(k)}(0)=0$ for all $k$ by giving an explicit error bound. Third, to connect the expansion to the adiabatic theorem, we show that
  \begin{equation}
  \min_\theta\|\psi(\epsilon,1)-e^{i\theta}\phi(1)\|\leq 2 \|\psi(\epsilon,1)-\sum_{j=0}^{M-1} \phi_j(s) \epsilon^j\|,
  \end{equation}
  for some $\theta$ for all $M$ if $G^{(k)}(1)=0$ for all $k$.   This already proves an error bound of $O(\epsilon^M)$ for any $M$. Finally, if $G$ is Gevrey class, then the error bound can be expressed with the help of the parameters appearing in the Gevrey condition and using a suitable choice of $M$ yields to the bound in Eq. (\ref{eq:exp_decay_error}).

  \paragraph*{Explicit form of $\phi_j$.}
 To satisfy the equation at $s=0$, we require $\phi_0(0)=\phi(0)$ and $\phi_j(0)=0$ for all $j>0$. Furthermore, substituting back the ansatz to the Schr\"{o}dinger equation Eq.(\ref{eq:schroedinger}), following \cite{Hagedorn2002}, we arrive at the recursion
  \begin{equation}
    \phi_j=\varphi_j \phi + i G^{-1} \dot{\phi}_{j-1},
   \quad
    \varphi_j=i\int_0^t\dd t' \expval{\dot{\phi}(t')}{G^{-1}(t')}{\dot{\phi}_{j-1}(t')} ,
    \label{eq:adiabatic_recursion}
  \end{equation}
  for all $j>0$, where $\varphi_j(s)$ is a complex number and $G^{-1}$ is the pseudo-inverse of $G$, and initial values are
  \begin{align}
    \phi_0(s)&= \phi(s)\\
    \varphi_0(s)&=1.
  \end{align}
  Note that $\varphi_j(0)$ has to be zero in order for $\phi_j(0)$ to be zero, but this is not a sufficient condition. We will investigate below when  $\phi_j(0)=0$ can be satisfied.
  
  \paragraph*{Existence of the expansion.}   To satisfy $\phi_j(0)=0$ for $j>0$, $\dot{\phi}_{j-1}(0)$ needs to be parallel to $\phi(0)$. This is satisfied if all derivatives of $G$ are $0$ at $s=0$ (see Lemma~\ref{lem:parallel} below). We show that if this condition is fulfilled, then the expansion exists.
  
  Define the truncation of the asymptotic series expansion,
  \begin{equation}
    \label{eq:truncation}
    \psi_M=\sum_{j=0}^{M-1} \phi_j \epsilon^j .
  \end{equation}
  Note that if $\|\psi-\psi_M\| = O(\epsilon^{M-1})$, then the expansion exists. Indeed, then $\|\psi-\psi_M\| = \|\psi-\psi_{M+1} + \epsilon^M \phi_M\| = O(\epsilon^M)$. By construction, $\psi_M$ almost satisfies the Schr\"odinger equation: 
  $i\epsilon \dot{\psi}_M=G\psi_M + i\epsilon^{M} \dot{\phi}_{M-1}$.
  Let $U(s)$ be the dynamics generated by $G/\epsilon$. Then, $\|\psi_M(\epsilon,s)-\psi(\epsilon,s)\|=\|U(s)^\dagger\psi_M(\epsilon,s)-\phi(0)\|$ and
    \begin{equation}\label{eq:1-0}
      \left\|U(s)^\dagger\psi_M(\epsilon,s)-\phi(0)\right\|=\left\|\int_0^s \dd s' \frac{\dd}{\dd s'}\left( U^\dagger \psi_M \right) \right\|= \left\|\epsilon^{M-1} \int_{0}^{s} \dd s' U^\dagger  \dot{\phi}_{M-1}\right\| \leq \epsilon^{M-1}\int_{0}^{s} \dd s' \left\|\dot{\phi}_{M-1}\right\|,
    \end{equation}
    where we used that  if the first $M$ derivatives of $G$ are $0$, then $\psi_M(\epsilon,0)=\phi(0)$. This proves the existence of the expansion.
    
  \paragraph*{Connecting the expansion to the adiabatic theorem.}
  Using the triangle inequality, we obtain
  \begin{equation}
    \label{eq:1}
    \min_\theta\|\psi(\epsilon,1)-e^{i\theta}\phi(1)\|\leq \|\psi(\epsilon,1) -\psi_M(\epsilon,1)\| +\min_\theta\|\psi_M(\epsilon,1)-e^{i\theta}\phi(1)\|.
  \end{equation}
  
  In Lemma~\ref{lem:parallel}, we prove that if the first $M$ derivatives of $G(s)$ vanish at $s=1$, then $\phi_j(1)$ is parallel to $\phi(1)$ for all $j=1,\ldots, M$. Therefore,  $\psi_M(\epsilon,1)$ is parallel to $\phi(1)$, so that $\min_\theta \|\psi_M(\epsilon,1)-e^{i\theta}\phi(1)\|=\left|\|\psi_M(\epsilon,1)\|-1\right|$. But $1=\|\psi(\epsilon,1)\|$, so using the triangle inequality, we get $\min_\theta \|\psi_M(\epsilon,1)-e^{i\theta}\phi(1)\|\leq \|\psi_M(\epsilon,1)-\psi(\epsilon,1)\|$. Therefore,
  \begin{equation}\label{eq:1-1}
    \min_\theta\|\psi(\epsilon,1)-e^{i\theta}\phi(1)\|\leq 2 \|\psi(\epsilon,1) -\psi_M(\epsilon,1)\|.
  \end{equation}
    
  \paragraph*{Choice of $M$.} From Eq.~\eqref{eq:1-0} and \eqref{eq:1-1}, 
  \begin{equation}
  \min_\theta\|\psi(\epsilon,1)-e^{i\theta}\phi(1)\|\leq 2 \|\psi(\epsilon,1)-\psi_M(\epsilon,1)\| \leq 2\epsilon^{M-1} \int_0^1 \|\dot\phi_{M-1}\|,
  \end{equation}
  so we only need to bound $\|\dot\phi_{M-1}\|$. We do this by using that $G$ is Gevrey class. From Lemma~\ref{lem:derivative_bound} below,
  \begin{equation}
    \|\dot\phi_{M-1}\|\leq 2\frac{4\pi^2}{3}\cdot 2c\frac{K}{\Delta}\left(\frac{4\pi^2}{3}\right)^2\cdot  \left[\frac{K^2}{\Delta^3}4c^2 \left(\frac{4\pi^2}{3}\right)^5 \right]^{M-1}\frac{[M!]^{1+\alpha}}{(M+1)^2}.
  \end{equation}
  Therefore,
  \begin{equation}
    \|\psi(\epsilon,1) - \phi(1)\|\leq \epsilon^{M-1} \cdot 8 c\frac{K}{\Delta}\left(\frac{4\pi^2}{3}\right)^3\cdot  \left[\frac{K^2}{\Delta^3}4c^2 \left(\frac{4\pi^2}{3}\right)^5 \right]^{M-1}\frac{[M!]^{1+\alpha}}{(M+1)^2}.
  \end{equation}
  Changing $M$ to $M+1$ and using that 
	${[(M+1)!]^{1+\alpha}}/{ (M+2)^2}\leq M^{(1+\alpha)M}$,
  we obtain
  \begin{equation}
    \|\psi(\epsilon,1) - \phi(1)\|\leq 8 c\frac{K}{\Delta}\left(\frac{4\pi^2}{3}\right)^3\cdot  \left[ \frac{K^2}{\Delta^3}4c^2 \left(\frac{4\pi^2}{3}\right)^5 \frac{M^{1+\alpha}}{\tau}\right]^{M}.
  \end{equation}
  This is true for any $M$, so setting 
  \begin{equation}
    M = \left\lfloor \left(\frac{\tau\Delta^3}{eK^24c^2 \left(\frac{4\pi^2}{3}\right)^5}\right)^{\frac{1}{1+\alpha}}\right\rfloor,
  \end{equation}
  we obtain 
    \begin{equation}
      \|\psi(\epsilon,1) - \phi(1)\|\leq 8 ce\frac{K}{\Delta}\left(\frac{4\pi^2}{3}\right)^3\cdot \exp\left\{- \left(\tau \frac{\Delta^3}{K^2}\frac{1}{4ec^2}\left(\frac{3}{4\pi^2}\right)^5\right)^{\frac{1}{1+\alpha}} \right\} .
    \end{equation}
  This proves Theorem~\ref{thm:AT_Gevrey}.
\end{proof}

We have repeatedly used the following lemma in the proof of Theorem~\ref{thm:AT_Gevrey}.

  \begin{lemma}\label{lem:parallel}
    If $G^{(k)}(s_0) = 0$   for some $s_0\in[0,1]$ and for all $k=1\dots M$, then  
    \begin{enumerate}[(i)]
      \item $\phi^{(k)}(s_0)$  is parallel to $\phi(s_0)$ for all $k=0,\ldots,M$,
      \item $\left[G^{(-1)}\right]^{(k)}(s_0) = 0$ for all $k=1,\ldots,M$,
      \item $\phi_k^{(l)}(s_0)$ is parallel to $\phi(s_0)$ for all $k=0,\ldots,M$ and $l=0,\dots,M-k$. 
    \end{enumerate}
  \end{lemma}
  \begin{proof}
    $G\phi=0$, so $(G\phi)^{(k)}=\sum_{j=0}^k\binom{k}{j} G^{(j)}\phi^{(k-j)}=0$ for all $k$. Applying this for $k\leq M$ and evaluating the result at $s_0$, the derivatives of $G$ vanish, thus $G\phi^{(k)}=0$ and therefore $\phi^{(k)}$ is parallel to $\phi$ at $s_0$, which proves (i).
    
    To prove (ii), use the Cauchy formula 
\begin{equation}\label{eq:parallel-cauchy}
    G^{-1}=\frac{1}{2\pi} \oint_\Gamma (z-G)^{-1}\frac{1}{z}\dd z,
\end{equation}    
    where $\Gamma=\left\{z\in \mathbb{C} \mid |z|=\Delta/2\right\}$ is a fixed curve around 0. Taking the $k$th derivative of Eq.~\eqref{eq:parallel-cauchy} and evaluating it at $s_0$, we see  that the derivatives of $G^{-1}$ also disappear.
     
	To prove (iii), we proceed by induction on $k$.     
     By (i), the claim is true for $k=0$. For $k>0$, we have $\phi_{k}^{(l)}=(\varphi_{k} \phi)^{(l)}+i(G^{-1}\phi_{k-1}^{(1)})^{(l)}$.
    By (i), the first term is parallel to $\phi$ at $s_0$. The second term consists of derivatives of $G^{-1}$ and derivatives of $\phi_k^{(1)}$. By (ii), the derivatives of $G^{-1}$ vanish at $s_0$, so that the only remaining term is $iG^{-1}\phi_{k-1}^{(l+1)}$. But this term also vanishes at $s_0$ by the induction hypothesis. This proves (iii).
	\end{proof}

In the remainder of this section, we derive the bound on the norm of $\dot{\phi}_{M-1}$ which was used in the proof of Theorem~\ref{thm:AT_Gevrey}, following the analysis in \cite{Nenciu1993}. First, we recall two technical lemmas from \cite{Nenciu1993}, which will be used repeatedly.

\begin{lemma}\label{lem:sum}
  Let $p,q$ be non-negative integers and  $r=p+q$. Then,
  \begin{equation}
    \sum_{l=0}^{k} {k\choose l} \frac{[(l+p)!(k-l+q)!]^{1+\alpha}}{(l+p+1)^2(k-l+q+1)^2}\leq \frac{[(k+r)!]^{1+\alpha}}{(k+r+1)^2} \frac{4\pi^2}{3}.
  \end{equation}
\end{lemma}
\begin{proof}
  Notice that if $r=p+q$, then
  \begin{equation}
    {k\choose l} [(l+p)!(k+q-l)!]^{1+\alpha}=\frac{{k \choose l}[(k+r)!]^{1+\alpha}}{{k+r\choose l+p}^{1+\alpha}}\leq [(k+r)!]^{1+\alpha}.
  \end{equation}
To upper-bound the summation, divide the sum into two parts at $\lfloor (k-p+q)/2 \rfloor$. If $l<\lfloor (k-p+q)/2 \rfloor$, then $(k-l+q+1)>\lfloor (k+p+q)/2 \rfloor+1$. Otherwise, if $l\geq \lfloor (k-p+q)/2 \rfloor$, then $(l+p+1)\geq \lfloor (k+p+q)/2 \rfloor+1$. Therefore,
  \begin{equation}
    \sum_{l=0}^k \frac{1}{(l+p+1)^2(k-l+q+1)^2}\leq 2 \sum_{l=0}^{\infty} \frac{1}{(l+1)^2}\frac{1}{(\lfloor (k+p+q)/2 \rfloor+1)^2}.
  \end{equation}
This can be upper-bounded by $4\pi^2/3$ as $k+p+q+1\leq 2(\lfloor (k+p+q)/2 \rfloor+1)$. This finishes the proof of Lemma~\ref{lem:sum}.
\end{proof}

We now use Lemma~\ref{lem:sum} to prove that if $A$ and $B$ are Gevrey-class, then their product is also Gevrey-class.
\begin{lemma}
  \label{lem:product}
  Let $A(s),B(s)$ ($s\in[0,1]$) be smooth paths of either vectors in  $\mathcal H$ or operators in $\mathcal{B}(\mathcal H)$ satisfying
  \begin{equation}\label{eq:product_assumption}
    \|A^{(k)}\|  \leq C d^k \frac{[(k+p)!]^{1+\alpha}}{(k+p+1)^2}\quad\text{ and } \quad
    \|B^{(k)}\|  \leq E f^k \frac{[(k+q)!]^{1+\alpha}}{(k+q+1)^2}
  \end{equation}
  for some non-negative constants $C,d,E,f$, non-negative integers  $p,q$, and for all $k\geq 0$. Then,
  \begin{equation}
    \|(AB)^{(k)}\|\leq \frac{4\pi^2}{3} CE g^k \frac{[(k+r)!]^{1+\alpha}}{(k+r+1)^2}
  \end{equation}
  for all $k\geq 0$, where $g=max(d,f)$ and $r=p+q$.
\end{lemma}
\begin{proof} We have
  \begin{equation}
    \|(AB)^{(k)}\|\leq \sum_l {k\choose l} \|A^{(l)}\| \|B^{(k-l)}\|,
  \end{equation}
  so inserting the bounds \eqref{eq:product_assumption} and upper-bounding $d$ and $f$ by $g$, we obtain
  \begin{equation}
    \|(AB)^{(k)}\|\leq CE g^k\sum_l {k\choose l}\frac{[(l+p)!]^{1+\alpha}}{(l+p+1)^2}\frac{[(k-l+q)!]^{1+\alpha}}{(k-l+q+1)^2}.
  \end{equation}
Using Lemma~\ref{lem:sum} to upper-bound the r.h.s. of this expression proves Lemma~\ref{lem:product}.
\end{proof}

Next we give a bound on the derivatives of the pseudo-inverse $G^{-1}$. As $G$ is non-invertible, the proof consists of two steps: first reducing the problem to the invertible case, then proving that the inverse of an invertible Gevrey-class operator is again Gevrey-class (assuming that the inverse is uniformly bounded).
\begin{lemma}\label{lem:inverse} If $G$ satisfies Eq. (\ref{eq:2}), then for all $k\geq 0$,
  \begin{equation}
    \|(G^{-1})^{(k)}\|\leq \frac{2}{\Delta} \left(\frac{K}{\Delta}2c\frac{4\pi^2}{3}\right)^k\frac{[k!]^{1+\alpha}}{(k+1)^2}.
  \end{equation}
\end{lemma}
\begin{proof}
  First, write the pseudo-inverse using the Cauchy formula,
  \begin{equation} \label{eq:inverse-cauchy}
    G^{-1}=\frac{1}{2\pi i} \oint_\Gamma \frac{1}{G-z}\frac{1}{z} \dd z,
  \end{equation}
  where $\Gamma=\{z\in \mathbb{C} \big||z|=\Delta/2\}$ is a fixed, $s$-independent curve. Taking the $k$th derivative of Eq. \eqref{eq:inverse-cauchy} (with respect to $s$), we get
  \begin{equation}
    (G^{-1})^{(k)}=\frac{1}{2\pi i} \oint_\Gamma [(G-z)^{-1}]^{(k)}\frac{1}{z} \dd z.
  \end{equation}
Thus, the norm of the pseudo-inverse can be bounded by the triangle inequality,
  \begin{equation}\label{eq:cauchy_bound}
    \|(G^{-1})^{(k)}\|\leq\max_{z\in \Gamma} \|[(G-z)^{-1}]^{(k)}\|.
  \end{equation}
 Note  that $G-z$ is invertible and $\|(G-z)^{-1}\|\leq 2/\Delta$ for $z\in \Gamma$. We now show that $(G-z)^{-1}$ is also Gevrey-class, more precisely that 
  \begin{equation}\label{eq:inverse-G-z}
    \left\|\left[(G-z)^{-1}\right]^{(k)}\right\|\leq \frac{2}{\Delta}\left(\frac{2}{\Delta}Kc\frac{4\pi^2}{3}\right)^k\frac{[k!]^{1+\alpha}}{(k+1)^2}
  \end{equation}
  for $k\geq 0$. To show this, we proceed by induction. For $k=0$, the bound obviously holds.  Taking the $k$th derivative of   $(G-z)(G-z)^{-1}=\Id$, 
   we get
  \begin{equation}
    \left[(G-z)^{-1}\right]^{(k)} = (G-z)^{-1} \sum_{l=1}^k {k \choose l} (G-z)^{(l)} \left[(G-z)^{-1}\right]^{(k-l)}.
  \end{equation}
    Using the induction hypothesis and collecting terms (notice that $l\geq 1$ and $k-l\leq k-1$), we get
    \begin{equation}\label{eq:inverse-sum}
      \left\|\left[(G-z)^{-1}\right]^{(k)}\right\|\leq \frac{2}{\Delta}\left(\frac{2}{\Delta}Kc\right)^k \left(\frac{4\pi^2}{3}\right)^{k-1}\sum_{l=1}^k {k \choose l} \frac{[l!(k-l)!]^{1+\alpha}}{(l+1)^2(k-l+1)^2}.
    \end{equation}
    Using Lemma \ref{lem:sum} to upper-bound the sum in \eqref{eq:inverse-sum}, we get
    \begin{equation}
      \left\|\left[(G-z)^{-1}\right]^{(k)}\right\|\leq \frac{2}{\Delta}\left(\frac{2}{\Delta}Kc \frac{4\pi^2}{3}\right)^{k}\frac{[k!]^{1+\alpha}}{(k+1)^2}.
    \end{equation}  
  This proves \eqref{eq:inverse-G-z}. Substituting this bound into Eq.~\eqref{eq:cauchy_bound} proves Lemma~\ref{lem:inverse}.
\end{proof}

Next, we prove that the ground state is also Gevrey-class (with the special choice of the phase as above).

\begin{lemma} \label{lem:gs_derivative}If $G$ satisfies Eq. (\ref{eq:2}), then the ground state $\phi$ satisfies 
  \begin{equation}\label{eq:gs-der-lem}
    \left\|\phi^{(k)}\right\|\leq  \left(2c\frac{K}{\Delta}\left(\frac{4\pi^2}{3}\right)^2\right)^{k} \frac{[k!]^{1+\alpha}}{(k+1)^2}.
  \end{equation}
  for all $k\geq0$, where $K,c$ and $\alpha$ are defined in Eq. (\ref{eq:2}) and $\Delta$ is the minimal gap of $G$.
\end{lemma}
\begin{proof}
  We  proceed by induction on $k$. For $k=0$, \eqref{eq:gs-der-lem} just reads $\|\phi\|\leq 1$. For $k>0$, notice that $G\phi=0$ and $\phi^{(1)}=-G^{-1}G^{(1)}\phi$ since the phase of $\phi$ is chosen such that $\braket{\dot{\phi}}{\phi}=0$. Therefore,
  \begin{equation}
    \left\|\phi^{(k)}\right\|=\left\|\left[G^{-1}G^{(1)}\phi\right]^{(k-1)}\right\|.
  \end{equation}
  Expanding the derivatives, we get 
  \begin{equation}\label{eq:gs-der-expanded}
    \left\|\phi^{(k)}\right\|\leq\sum_{l=0}^{k-1} {k-1\choose l}\left\|\left[G^{-1}G^{(1)}\right]^{(k-l-1)}\right\|\left\|\phi^{(l)}\right\|.
  \end{equation}
  The right hand side can be bounded using the induction hypothesis as the higest derivative of $\phi$ appearing there is the $(k-1)$th. For that, we first derive a bound on the norm of the derivatives of $G^{-1}G^{(1)}$. This can be done by applying Lemma \ref{lem:product} to $G^{(1)}$ and $G^{-1}$ and using  Lemma \ref{lem:inverse} to obtain
  \begin{equation}
    \left\|[G^{-1}G^{(1)}]^{(k)}\right\|\leq\left(Kc\frac{2}{\Delta}\frac{4\pi^2}{3}\right)^{k+1} \frac{[(k+1)!]^{1+\alpha}}{(k+2)^2}
  \end{equation}
  for $k\geq 0$. Substituting this bound into \eqref{eq:gs-der-expanded}, we obtain
  
    \begin{equation}
      \|\phi^{(k)}\|\leq\sum_{l=0}^{k-1} {k-1\choose l}\left(Kc\frac{2}{\Delta}\frac{4\pi^2}{3}\right)^{k-l} \frac{[(k-l)!]^{1+\alpha}}{(k-l+1)^2}\left(Kc\frac{2}{\Delta}\left(\frac{4\pi^2}{3}\right)^2\right)^{l} \frac{[l!]^{1+\alpha}}{(l+1)^2}.
    \end{equation}
Notice that $l\leq k-1$, so 
    \begin{equation}
      \|\phi^{(k)}\|\leq\left(Kc\frac{2}{\Delta}\frac{4\pi^2}{3}\right)^{k}\left(\frac{4\pi^2}{3}\right)^{k-1}\sum_{l=0}^{k-1} {k-1\choose l} \frac{[(k-l)!]^{1+\alpha}}{(k-l+1)^2}\frac{[l!]^{1+\alpha}}{(l+1)^2}.
    \end{equation}
Thus, using Lemma \ref{lem:sum}, we obtain
  \begin{equation}
    \left\|\phi^{(k)}\right\|\leq \left(Kc\frac{2}{\Delta}\left(\frac{4\pi^2}{3}\right)^2\right)^{k} \frac{[k!]^{1+\alpha}}{(k+1)^2},
  \end{equation}
  which proves Lemma~\ref{lem:gs_derivative}.
\end{proof}

We are now in the position to bound $\|\dot\phi_{M-1}\|$. Instead of bounding it directly, we prove a general bound on all $\|\phi_{j}^{(k)}\|$. The desired bound is obtained then by setting $j=M-1$ and $k=1$.
\begin{lemma}\label{lem:derivative_bound} For all $j,k\geq 0$, the vectors $\phi_j$ and scalars $\varphi_j$ defined in Eq. (\ref{eq:adiabatic_recursion}) satisfy 
  \begin{equation} \label{eq:derivative_bound}
    \|\phi_j^{(k)}\|\leq A_1 A_2^j A_3^k \frac{[(k+j)!]^{1+\alpha}}{(k+j+1)^2}
\quad\text{ and } \quad
    |\varphi_j^{(k)}|\leq A_2^j A_3^k \frac{[(k+j)!]^{1+\alpha}}{(k+j+1)^2},
  \end{equation}
  where the constants $A_1$, $A_2$ and $A_3$ can be expressed with the constants appearing in Eq. (\ref{eq:2}):
\begin{equation}
  A_1  = \frac{8\pi^2}{3},\quad
  A_3 = 2c\frac{K}{\Delta}\left(\frac{4\pi^2}{3}\right)^2,\quad
  A_2 = 4c^2\frac{K^2}{\Delta^3}\left(\frac{4\pi^2}{3}\right)^5.
\end{equation}
\end{lemma}
\begin{proof} 
  We proceed by induction on $j$, using the recursion  in relation \eqref{eq:adiabatic_recursion}.
  We first bound $|\varphi_j|$ 
  using the induction hypothesis, then bound $|\varphi_j^{(k)}|$ for $k>0$ before bounding $\|\phi_j^{(k)}\|$. 

\paragraph*{Base case.}
We have $\varphi_0(s)=1$ and $\phi_0(s)=\phi(s)$, 
so \eqref{eq:derivative_bound} holds for $j=0$ since
\begin{equation}
  A_1\geq 1\quad\text{ and }\quad
  A_3\geq 2c\frac{K}{\Delta}\left(\frac{4\pi^2}{3}\right)^2.
\end{equation}

\paragraph*{Bound on $|\varphi_j|$, $j\geq 1$.}
 $|\varphi_j|$ can be bounded by the maximum value of the integrand in Eq. \eqref{eq:adiabatic_recursion},
 \begin{equation}
   |\varphi_j|\leq \|G^{-1}\dot\phi\| \|\dot\phi_{j-1}\|\leq \|G^{-1}\|\cdot \|\dot\phi\| \cdot \|\dot\phi_{j-1}\|.
 \end{equation}
 Using the bound on $\|\dot{\phi}\|$ from Lemma~\ref{lem:gs_derivative}, $\|G^{-1}\|$ from Lemma~\ref{lem:derivative_bound} and the induction hypothesis on  $\|\dot\phi_{j-1}\|$, 
 we get
\begin{equation}
  |\varphi_j|\leq \frac{2}{\Delta} \cdot A_3 \frac{1}{4} \cdot A_1 A_2^{j-1}A_3  \frac{[(j)!]^{1+\alpha}}{(j+1)^2} \leq A_2^j \frac{[(j)!]^{1+\alpha}}{(j+1)^2}
\end{equation}
since
\begin{equation}
  1\geq A_1\frac{ A_3^2 }{A_2}\frac{2}{\Delta} \frac{1}{4}.
\end{equation}

\paragraph*{Bound on $\|\varphi_j^{(k)}\|$.}
We now
bound $|\varphi_j^{(k)}|$ for $k > 0$.  
First, from the induction hypothesis,
\begin{equation} 
  \|\dot\phi_{j-1}^{(k)}\|\leq A_1 A_2^{j-1}A_3^{k+1}\frac{[(k+j)!]^{1+\alpha}}{(k+j+1)^2} .
\end{equation}
Then, using Lemma \ref{lem:product} and Lemma \ref{lem:inverse}, 
we get that for all $k\geq 0$,
  \begin{equation}\label{eq:bound_intermediate}
    \left\|\left(G^{-1}\dot{\phi}_{j-1}\right)^{(k)}\right\|\leq\frac{4\pi^2}{3}\frac{2}{\Delta}A_1 A_2^{j-1}A_3^{k+1} \frac{[(k+j)!]^{1+\alpha}}{(k+j+1)^2}.
  \end{equation}
Moreover, from Lemma~\ref{lem:gs_derivative},
\begin{equation}\label{eq:bound_intermediate2}
  \|\dot\phi^{(k)}\|\leq A_3^{k+1}\frac{[(k+1)!]^{1+\alpha}}{(k+2)^2}.
\end{equation}
Since $\dot\varphi_j= \braket{\dot\phi}{G^{-1}\dot\phi_{j-1}}$, Lemma~\ref{lem:product}, Eq.~\eqref{eq:bound_intermediate} and \eqref{eq:bound_intermediate2} imply that 
\begin{equation}
 |\varphi_j^{(k+1)}|=|\dot \varphi_j^{(k)}|\leq \left(\frac{4\pi^2}{3}\right)^2\frac{2}{\Delta} A_3 A_1A_2^{j-1}A_3^{k+1}\frac{[(k+j+1)!]^{1+\alpha}}{(k+j+2)^2}
\end{equation}
for $k\geq 0$. Changing $k+1$ to $k$ gives
\begin{equation}
  |\varphi_j^{(k)}|\leq \left(\frac{4\pi^2}{3}\right)^2\frac{2}{\Delta} A_3 A_1A_2^{j-1}A_3^{k}\frac{[(k+j)!]^{1+\alpha}}{(k+j+1)^2}\leq  A_2^j A_3^k \frac{[(k+j)!]^{1+\alpha}}{(k+j+1)^2},
\end{equation}
 since
$
  1\geq A_1 \frac{A_3 }{A_2}\frac{2}{\Delta} \left(\frac{4\pi^2}{3}\right)^2
$.

\paragraph*{Bound on $\|\phi_j^{(k)}\|$.}
We now bound $\|\phi_j^{(k)}\|$. 
 Using the bound on $\|\phi^{(k)}\|$ from Lemma~\ref{lem:gs_derivative} and $|\varphi_j^{(k)}|$, Lemma \ref{lem:product} implies
  \begin{equation}
    \|(\varphi_j \phi)^{k}\|\leq \frac{4\pi^2}{3}A_2^{j} A_3^k \frac{[(k+j)!]^{1+\alpha}}{(k+j+1)^2}. 
  \end{equation}
Finally, using  Eq.~\eqref{eq:adiabatic_recursion},
\begin{equation}
  \|\phi_j^{(k)}\| \leq 2 \max\left(\left\|(\varphi_j \phi)^{(k)}\right\|, \left\|\left(G^{-1}\dot{\phi}_{j-1}\right)^{(k)}\right\|\right).
\end{equation}
Hence, since 
\begin{equation}
  2\frac{4\pi^2}{3}  \leq A_1\quad\text{ and }\quad
  \frac{4\pi^2}{3}\frac{2}{\Delta} A_3 \leq A_2,
\end{equation}
we obtain 
\begin{equation}
   \|\phi_j^{(k)}\|\leq A_1 A_2^j A_3^k \frac{[(k+j)!]^{1+\alpha}}{(k+j+1)^2},  
\end{equation}
which finishes the proof of Lemma~\ref{lem:derivative_bound} and hence the proof of Theorem~\ref{thm:AT_Gevrey}.
\end{proof}

\section{Locality of local adiabatic change}

We show in this section that $G_n$ and $\tilde{G}_n$, as defined in Eq. (\ref{eq:localisedG}) and in Eq. (\ref{eq:entireG}) in the main text, generate basically the same dynamics.
The proof relies on $\tilde{G}(0)$ being frustration-free, and a runtime of $\tau =O(\log^{1+\alpha}(N/\ve))$, because it turns out that the achieved locality scales linearly with the runtime. We also use the Lieb-Robinson bound \cite{lieb1972, HastingsKoma2006,NachtergaeleSims2006,BravyiHastingsVerstraete2006}, which states that if $H$ is a local (possibly time-dependent) Hamiltonian with uniformly bounded interaction strengths, $U(t)$ is the unitary evolution generated by $H$, and $O_A,O_B$ are operators supported on regions $A,B$, respectively, then 
\begin{equation}\label{eq:liebRobinson}
	\| [U(t)O_A U^\dagger(t), O_B] \| \leq c\min(|A|,|B|)\|O_A\|\|O_B\|\exp\left(\gamma t - \nu L\right),
\end{equation}
where $L$ is the distance between $A$ and $B$, and $c,\gamma,\nu$ are constants depending only on the geometry of the lattice and the maximum interaction strength.

The following theorem justifies the  replacement  of \eqref{eq:entireG} with  \eqref{eq:localisedG} in the main text, without significantly altering the evolution and thus the final state. 

\begin{theorem}\label{thm:locality}
	Let $\tilde G(s) = \sum_{\mu\in \Mu} G_\mu(s)$ be a frustration-free Hamiltonian path with $O(N)$ local terms such that $|\supp \frac{\dd}{\dd s}\tilde G| = O(1)$, and let $G$ be a localised version of $\tilde G$, i.e.,
	\begin{equation}\label{eq:SM_localized_G}
  G(s) = \sum_{\mu\in \Mu'} G_\mu(s), \quad \Mu' = \left\{\mu\in\Mu \mid d\Big(\supp\frac{\dd}{\dd s}{\tilde G},\supp G_\mu\Big)< \chi \tau \right\}
\end{equation}
for some constant $\chi$ and 
 adiabatic runtime $\tau$. Let $\psi$ and $\tilde \psi$ be the evolved states under $G$ and $\tilde G$ respectively, i.e.,
\begin{equation}
i\frac{\dd}{\dd t} \psi(t) = G\left(\frac{t}{\tau}\right) \psi(t),\quad
i\frac{\dd}{ \dd t} {\tilde \psi}(t) = \tilde G\left(\frac{t}{\tau}\right) \tilde \psi(t), \quad t\in[0,\tau],\quad
\psi(0)=\tilde \psi(0) = \phi(0),
\end{equation}
where $\phi(0)$ is the ground state of $\tilde G(0)$. 
Then, for sufficiently large $\chi=O(1)$,
\begin{equation}
	\|\tilde \psi(\tau)-\psi(\tau)\|\leq cN^2\tau^2e^{(\gamma-v\chi/2)\tau},
\end{equation}
where $c,\gamma,v$ are the constants from \eqref{eq:liebRobinson}. In particular, if $\tau = O(\log^{1+\alpha}(N/\ve))$, then
\begin{equation}
    \|\tilde \psi(\tau)-\psi(\tau)\|\leq \varepsilon /N
  \end{equation}
   for sufficiently large $\chi=O(1)$.
\end{theorem}

\begin{proof}
 For any $\Omega \subseteq \Mu$, let $G_\Omega = \sum_{\mu\in\Omega} G_\mu$ and $U_\Omega(t,s)$ be the unitary dynamics generated by $G_\Omega$. Then $U_\Omega$ satisfies
 
 \begin{align}
   \partial_t U_\Omega(t,s) & = - i G_\Omega(t/\tau) U_\Omega(t,s), \label{eq:partialtU}\\
   \partial_s U_\Omega(t,s) & = - i  U_\Omega(t,s) G_\Omega(s/\tau),\\
   U_\Omega(t,s) & = \mathcal{T}\exp \left\{-i\int\limits_{s}^{t}dt' G_\Omega(t'/\tau)\right\}.
 \end{align}
Notice that $G_{\Mu'} = G$ and $G_\Mu = \tilde G$. We write $U=U_{\Mu'}$ and $\tilde U= U_\Mu$.
 Let $B$ be the boundary of $\Mu'$, that is, $B=\{\mu\in\Mu \mid d(\lambda,\mu)= \lceil\chi \tau \rceil\}$ and $\bar{B}=\Mu \backslash B$.  Then, since $\tilde{G}(0)$ is frustration-free and  all terms outside of $\Mu'$ are constant, $U_{\bar{B}}(t,0) \phi(0) = U(t,0) \phi(0)$ as $U_{\bar{B}} = U\otimes U_{\bar B\backslash \Mu'}$ and $U_{\bar B\backslash \Mu'}\phi (0)= \phi(0)$. In other words, $G_{\bar{B}}$ generates the same dynamics as $G$. 
Thus,
\begin{equation}\label{eq:locality_unitary1}
  \|\tilde{\psi}(\tau)-\psi(\tau)\|=\|\tilde U\phi(0)-U\phi(0)\|=\|\tilde U\phi(0)-U_{\bar{B}}\phi(0)\|=\|\phi(0)-\tilde U^\dagger U_{\bar{B}} \phi(0)\|,
\end{equation}
where $\tilde U$ and $U_{\bar B}$ are evaluated at $(\tau,0)$.
Let $V(t)=\tilde U^\dagger(t,0) U_{\bar{B}}(t,0)$. Then, since $ G_B = \tilde G - G_{\bar B}$, Eq.~\eqref{eq:partialtU} implies 
\begin{equation}
  \frac{\dd}{\dd t} V = i \tilde U^\dagger(t,0) G_B(t/\tau) \tilde U(t,0) V(t).
\end{equation}
We now approximate $V$ with a local unitary to obtain a bound for \eqref{eq:locality_unitary1}. Let $X=\{\mu\in\Mu\mid d(\mu,B)\leq r\}$ for some $r$ to be specified below, and let 
\begin{equation}
  V_X(t)=\mathcal{T}\exp\left\{i\int_0^t dt' U^\dagger_X(t',0) G_{B}(t'/\tau) U_X(t',0)\right\}
\end{equation}
For $r =\frac12   \chi \tau$ and sufficiently large $\chi=O(1)$, 
$X$ and $\supp \dot G$ are disjoint since $|\supp \dot G|=O(1)$, so $G_X(t/\tau)=G_X(0)$ for all $t\in[0,\tau]$. Because of frustration-freeness, $G_X(t/\tau)\phi(0)=0$, and thus the dynamics generated by $G_X$ acts trivially on the initial state, i.e., $U_X(t,0)\phi(0)=\phi(0)$. Thus, $V_X$ also acts trivially on the initial state, $V_X(t)\phi(0)=\phi(0)$. Hence, substituting this into Eq.~\eqref{eq:locality_unitary1}, we get
\begin{equation}
  \| \phi(0)-\tilde U^\dagger U\phi(0)\|=\|V_X\phi(0)-\tilde U^\dagger U_{\bar{B}} \phi(0)\|=\|\phi(0)-V_X^\dagger V \phi(0)\|.
\end{equation}
From the definition of $V$ and of $V_X$, 
\begin{equation}\label{eq:dtVXV}
\frac{d}{dt}(V_X^\dagger V)={i}V_X^\dagger (\tilde U^\dagger G_B\tilde U-U_X^\dagger G_B U_X)V,
\end{equation}
where $G_B$ is evaluated at $t/\tau$. 
Thus, by integrating \eqref{eq:dtVXV} and using the triangle inequality and unitary invariance of the operator norm,
\begin{equation}
  \|\phi(0)-V_X^\dagger V \phi(0)\|\leq \int_0^\tau\dd t \|\tilde U^\dagger G_B\tilde U-U_X^\dagger G_B U_X\| = \int_0^\tau \dd t \|G_B-\tilde U U_X^\dagger G_B U_X \tilde U^\dagger \|,
\end{equation}
where the unitary evolutions are taken from $0$ to $t$ and $G_B$ is evaluated at $t/\tau$. Observe that 
\begin{equation} \label{eq:dsUUX}
  \partial_s \left(\tilde U(t,s)U^\dagger_X(t,s)\right)=-i\tilde U(t,s)G_{\bar{X}}(s/\tau)U^\dagger_X(t,s),
\end{equation}
where $\bar X = \Mu\backslash X$. Integrating \eqref{eq:dsUUX} over $s$, we get 
\begin{equation}
  G_B(t/\tau)-\tilde U U_X^\dagger G_B(t/\tau) U_X \tilde U^\dagger = {-}i\int_0^t \dd s \tilde U(t,s)\left[G_{\bar{X}}(s/\tau),U^\dagger_X(t,s) G_B(t/\tau) U_X(t,s)\right]\tilde U^\dagger(t,s).
\end{equation}
Therefore, using the triangle inequality and the unitary invariance of the norm, we get
\begin{align}
 \|\psi(\tau) - \tilde \psi(\tau)\|&\leq \int_0^\tau \dd t \int_0^t \dd s\left\|\left[U_{X}^\dagger(t,s) G_B(t/\tau) U_{X}(t,s),G_{\bar{X}}(s/\tau)\right]\right\| \\
  &\leq \int_0^\tau \dd t \int_0^t ds\ c N^2 e^{\gamma (t-s) - \nu r}\leq c N^2 \tau^2 e^{\gamma \tau -\nu r},
\end{align}
where the second line follows from the Lieb-Robinson bound as $B$ and $\bar{X}$ are separated by a distance $r=\frac12   \chi \tau$. 
This proves Theorem~\ref{thm:locality}.
\end{proof}

\section{Relaxations on the assumption of a uniform gap along the path}
In this section, we show that the assumption of a spectral gap along the entire path of $\tilde G_n$ can be relaxed.  

\begin{theorem}
	Suppose that $\tilde G_n(0)$ has a spectral gap of at least $\delta>0$. Then $\tilde G_n(s)$ has a spectral gap of at least $ q_0^{2}\delta$ for all $s\in[0,1]$, where $q_0$ satisfies that $\Id\geq Q_\lambda\geq q_0 \Id$ (with $Q_\lambda$  as in Eq. (\ref{PsiN}) in the main text). In particular, a uniform gap as defined in the main text implies a constantly lower bounded gap along the entire Hamiltonian path in the given algorithm.
\end{theorem}
\begin{proof}
Since $\tilde G_n(0)$ is positive semidefinite and has a non-trivial kernel, the spectral gap condition of $\tilde G_n(0)$ is equivalent to $G_n(0)^2\geq \delta G_n(0)$.
 Note that $\Id\geq A_{n,m}(s)\geq q_0\Id$ (with $A_{n,m}$ as in Eq.~\eqref{eq:Anlambda}).
 Let $X_n(s)=q_0^{2}A_{n,n}^{-1}(s)\tilde G_n(0)A_{n,n}^{-1}(s)$. Then,
\begin{align}
 \tilde G_n(s) &= \sum_{\mu\in\Mu} \left(\prod_{m:\lambda_m\in \Lambda_\mu}  A_{n,m}^{-1}(s) \right)P_\mu \left(\prod_{m:\lambda_m\in \Lambda_\mu}   A_{n,m}^{-1}(s)\right)\\
 &\geq q_0^{2}A_{n,n}^{-1}(s) \sum_{\mu\in\Mu} \left(\prod_{\substack{m: \lambda_m\in \Lambda_\mu\\ m\neq n}}  A_{n,m}^{-1}(s) \right)P_\mu \left(\prod_{\substack{m:\lambda_m\in \Lambda_\mu\\ m\neq n}}   A_{n,m}^{-1}(s)\right)A_{n,n}^{-1}(s)\\
 &= q_0^{2}A_{n,n}^{-1}(s)\tilde G_n(0)A_{n,n}^{-1}(s)\\ &=X_n(s),
\end{align}
where we used in the second line $\|A_{n,n}^{-1}(s)\|\leq q_0^{-1}$. 
Notice that $\tilde G_n(s)$ and $X_n(s)$ have the same kernel and are both positive semidefinite. Thus, the gap of $\tilde G_n(s)$ is lower bounded by the gap of $X_n(s)$. But since $G_n(0)^2\geq \delta G_n(0)$, we also have $X_n(s)^2 \geq q_0^{2}\delta X_n(s)$ since $A_{n,n}^{-2}(s)\geq \id$. Thus, $\tilde G_n(s)$ has a spectral gap of at least $\Delta_n\geq q_0^{2}\delta$. 
\end{proof}

\section{Gibbs state preparation in the non-commuting case for high temperatures}

The algorithm we presented to prepare a purifiaction of the Gibbs state of a Hamiltonian used explicitly that the Hamiltonian is a sum of commuting terms. Thus, one may wonder if one can apply it directly to Gibbs states of non-commuting Hamiltonians $H$. The genaral answer is no. The reason is that even though a parent Hamiltonian can still be defined as
\begin{equation}\label{eq:Gnl}
 G^{nl}(\beta)  = \sum_{\mu\in\Mu} G^{nl}_\mu(\beta) = \sum_{\mu\in\Mu} e^{\frac{\beta}{2} H} P_\mu e^{-  \beta H} P_\mu e^{\frac{\beta}{2} H},
\end{equation}
now the terms are not local (hence the superscript \emph{nl}), and the norm of each term may be exponentially large in $N$. Thus, adiabatic state preparation using \eqref{eq:Gnl} directly takes exponential time. However, in this section we  show that for sufficiently high, but constant temperatures, one can approximate $G^{nl}$ by an $r=O(\log N)$-local Hamiltonian $G^r$ which is a sum of $O(\poly(N))$ terms. We also show that in this case, $G^{nl}$ (and thus also $G^r$) has a $\Omega(1)$ spectral gap and $O(N)$ norm. Because of the existence of the gap, the ground state of $G^r$ is a good approximation of the ground state of $G^{nl}$.

Using the adiabatic theorem, the following algorithm runs in $O(\poly(N))$ time for high enough (but $\Omega(1)$) temperatures and gives a good approximation of the purification of the Gibbs state of a non-commuting Hamiltonian:
\begin{enumerate}
  \item Prepare the ground state of $G^{r}(0)=\sum_{\mu\in\Mu} P_\mu$
  \item Calculate $G^r(\beta)$
  \item Prepare adiabatically the ground state of $G^r(\beta)$
\end{enumerate}

We first use the cluster expansion \cite{Kotecky1986} 
to construct the approximating Hamiltonian $G^r$. We also show that the norm of $G^{nl}$ is $O(N)$. Finally, we show that the gap of $G^{nl}$ is $\Omega(1)$. For simplicity, assume that $H=\sum_{\lambda\in\Lambda}$ is a sum of nearest-neighbour interactions, although the results and proofs generalise to other types of interactions. We also assume that $\|h_\lambda\| \leq 1$. 

\paragraph*{Cluster expansion.}

We now show that $G^{nl}$ can be approximated by an $r=O(\log N)$-local Hamiltonian $G^r$. More precisely, we show the following result.

\begin{theorem} \label{thm:clusterLogLocal}
	For sufficiently small (but constant) $\beta$, there exists an $r=O(\log N)$-local Hamiltonian $G^r$ with $O(\poly( N))$ terms such that
	\begin{equation}
		\|G^{nl} - G^r\| < O(1/\poly(N)).
	\end{equation}
	Moreover, the terms of $G^r$ can be calculated in $O(\poly(N))$ time. 
\end{theorem}

For any function $f$ defined on the subsets of $\Lambda$, define the \emph{M\"obius transformations}
  \begin{align}
    \hat{f} ( \Omega) & : =  \sum_{\Theta\subseteq\Omega} f ( \Theta), \label{eq:hat}\\
    \check{f} ( \Omega) & : =  \sum_{\Theta\subseteq\Omega} ( - 1)^{| \Omega \backslash  \Theta |} f( \Theta) \label{eq:check} .
  \end{align}
It is straightforward to check that the following Lemma holds \cite{[{See also }] [{}] PhysRevB.91.045138}.
\begin{lemma}[M{\"o}bius inversion] \label{lem:Mobius}
  \begin{equation}
  	 \hat{\check{f}} = \check{\hat{f}} = f.
\end{equation}  
\end{lemma}

For any $\Omega\subseteq\Lambda$, let $H_\Omega = \sum_{\lambda\in \Omega} h_\lambda$, and let $f_\mu ( \Omega) = e^{\beta H_\Omega} P_\mu e^{-  2\beta H_\Omega} P_\mu e^{\beta H_\Omega}$ for any $\mu\in \Mu$. 
Using Lemma~\ref{lem:Mobius}, one can express $G_\mu^{nl}$ as
\begin{equation}\label{eq:clusterexp}
G_\mu^{nl} = f_\mu ( \Lambda) = \sum_{\Omega \subseteq \Lambda} \check{f}_\mu( \Omega).
\end{equation}
This so-called \emph{cluster expansion} has many interesting properties. 

\begin{lemma}\label{lem:clusterzero}
Let $\mu\in\Mu$. 
If $\Omega\subseteq\Lambda$ is such that $\mu$ is disjoint from $\Omega$ and $\Theta\subseteq \Lambda$ is disjoint from $\Omega$, then 
\begin{equation}
  \check{f}_\mu ( \Omega \cup \Theta) =0.
\end{equation}
\end{lemma}
\begin{proof}
We have
	\begin{equation}
  \check{f}_\mu ( \Omega \cup \Theta) = \sum_{\Omega' \subseteq \Omega, \Theta' \subseteq \Theta} ( - 1)^{| \Omega\backslash \Omega' |} ( - 1)^{| \Theta\backslash \Theta' |} e^{\beta H_{\Theta'}} P_\mu e^{- 2\beta H_{\Theta'}} P_\mu e^{\beta H_{\Theta'}} =0,
	\end{equation}
since $H_{\Omega'}$ commutes with $P_\mu$ and with $H_{\Theta'}$, and the sum over $\Omega'$ is 0.
\end{proof}
  Lemma~\ref{lem:clusterzero} states that $\check{f}_\mu$ is non-zero only for connected sets of edges that, in addition, contain $\mu$. 
Another interesting property of $\check{f}_\mu$ is that its norm can be bounded as follows. 
\begin{lemma} For any $\Omega\subset\Lambda$ and any edge $\mu$,
 \begin{equation}
 	\|\check f_\mu(\Omega)\| \leq (e^{4\beta} -1 )^{|\Omega|}.
 \end{equation}
\end{lemma}
\begin{proof}
Expanding the exponentials, one gets
\begin{equation}\label{eq:checkf1}
 \check{f}_\mu ( \Omega) = \sum_{\Theta \subseteq \Omega} ( - 1)^{| \Omega\backslash \Theta |}    \sum_{w_1, w_2, w_3 \in \Theta^{\ast}} \frac{( - \beta)^{| w_1 |} \cdot ( 2   \beta)^{| w_2 |} \cdot ( - \beta)^{| w_3 |}}{| w_1 | ! \cdot | w_2 | !   \cdot | w_3 | !} h_{w_1} P_\mu h_{w_2} P_\mu h_{w_3} ,
\end{equation}
where $\Theta^*$ is the set of all finite sequences of elements of $\Theta$,  and for any $w\in\Theta^*$, $|w|$ denotes the length of $w$ and $h_w = h_{\lambda_1}\ldots h_{\lambda_{|w|}}$ if $w = (\lambda_1,\ldots,\lambda_{|w|})$.

Consider  the set $A=\mathrm{supp}(w_1) \cup \mathrm{supp}(w_2) \cup\mathrm{supp} (w_3)$. If $A\neq\Omega$, then the alternating sum in \eqref{eq:checkf1} over all $\Theta$ such that $A\subseteq\Theta\subseteq\Omega$ is 0. Thus, 
\begin{equation}
 \| \check{f}_\mu ( \Omega) \| \leq \sum_{{\begin{array}{c}
     w_1, w_2, w_3 \in \Omega^{\ast}\\
     {\supp} (w_1) \cup {\supp} (w_2) \cup {\supp} (w_3) = \Omega
   \end{array}}} \frac{\beta^{| w_1 |} \cdot ( 2 \beta)^{| w_2 |} \cdot
   \beta^{| w_3 |}}{| w_1 | ! \cdot | w_2 | ! \cdot | w_3 | !}.
\end{equation}
But this is exactly a M{\"o}bius transform, so
\begin{equation}
 \| \check{f}_\mu ( \Omega) \| \leq \sum_{\Theta \subseteq \Omega} ( - 1)^{| \Omega\backslash \Theta |}  \sum_{w_1, w_2, w_3 \in \Theta^{\ast}} \frac{\beta^{| w_1 |} \cdot ( 2 \beta)^{|w_2 |} \cdot \beta^{| w_3 |}}{| w_1 | ! \cdot | w_2 | ! \cdot | w_3 | !} = \sum_{\Theta \subseteq \Omega} ( - 1)^{| \Omega\backslash \Theta |} e^{4 \beta | \Theta |} = (e^{4 \beta} - 1)^{| \Omega |}.
\end{equation}
\end{proof}

We are now in a position to prove Theorem~\ref{thm:clusterLogLocal}.

\begin{proof}[Proof of Theorem~\ref{thm:clusterLogLocal}]
Using \eqref{eq:clusterexp}, we can write $G_\mu^{nl}$ as a sum of local terms where the norm of the terms decay exponentially with their support. As the number of terms with a given size is bounded by the lattice growth constant \cite{Klarner1967},  $\|G_\mu^{nl}\| = O(1)$ above some temperature. Indeed, let $\eta$ be the lattice growth constant, so that the number of sets of connected edges containing $\mu$ and of size $M$ is bounded by $e^{\eta M}$. Then,

\begin{equation}
\|G_\mu^{nl}\| \leq \sum_{\Omega \subseteq \Lambda} \|\check{f}_\mu( \Omega)\| \leq \sum_{M\geq 0} e^{\eta M} (e^{4\beta} -1)^M=O(1)
\end{equation}
if $\beta$ is sufficiently small (but constant). 
In this case, $G_\mu^{nl}$ can be approximated by an $r$-local operator $G^r_\mu$ by omitting all connected sets $\Omega$ of size at least $r$. The error of this approximation is 
\begin{equation}\label{eq:clustertrunc_error}
\|G^{nl}_\mu-G^r_\mu\| \leq \sum_{|\Omega|\geq r} \|\check{f}_\mu( \Omega)\| \leq \sum_{M\geq r} e^{\eta M} (e^{4\beta} -1)^M = \frac{\left(e^\eta \left(e^{4\beta}-1\right)\right)^r}{1-e^\eta \left(e^{4\beta}-1\right)} = \frac{y^r}{1-y},
\end{equation}
where $ y=e^\eta \left(e^{4\beta}-1\right)$. 
Therefore, above some constant temperature, the cluster expansion can be truncated at $|\Omega|\leq O(\log N)$, giving an error of $O(1/\poly(N))$. This results in a $O(\log
(N))$-local Hamiltonian
\begin{equation}
G^r =  \sum_\mu \sum_{|\Omega|\leq O(\log N)} \check{f}_\mu(\Omega) 
\end{equation}
 with $O(\poly ( N))$ terms. Note that this Hamiltonian can now be calculated in $O(\poly( N))$ time.
Indeed, there are $O(\poly (N))$ terms $\check{f}_\mu (\Omega)$, and each term can
be evaluated in $O(\poly (N))$ time since there are at most
$O(\poly (N))$ subsets of each $\Omega$.
\end{proof}

\paragraph*{Gap of $G^{nl}$.}

It remains to be shown that at sufficiently high (but constant) temperatures, the parent Hamiltonian is gapped.

\begin{theorem}\label{thm:gap}
 For sufficiently small (but constant) $\beta$, $G^{nl}$ has a spectral gap of $\Omega(1)$. 
\end{theorem}
\begin{proof}
To show the existence of a gap, we use that $G^{nl}\geq0$ is frustration-free, so it is enough to show that 
\begin{equation}\label{eq:square_gap_cond}
 (G^{nl})^2 \geq \Delta G^{nl} 
\end{equation}
for some $\Delta=\Omega(1)$. 
Expanding $G^{nl}$, \eqref{eq:square_gap_cond} is equivalent to
\begin{equation}
  \sum_\mu (G_\mu^{nl})^2 + \sum_{\mu\neq \nu} G_\mu^{nl} G_\nu^{nl} \geq \sum_\mu\Delta G_\mu^{nl}.
\end{equation}
Using Eq. (\ref{eq:clustertrunc_error}) with $r=1$, we get that $G_\mu^{nl}$ is close to $P_\mu=G_\mu^{r=1}$ for high temperatures and thus it is gapped and the gap is close to 1. Therefore, it is enough to show that for some other constant $\Delta'<1$,
\begin{equation}
  \sum_{\mu\neq \nu} G_\mu^{nl} G_\nu^{nl} \geq -\sum_\mu\Delta' G_\mu^{nl}.
\end{equation}
We upper bound the r.h.s. by lower bounding $\sum_\mu G_\mu^{nl}$ as
\begin{equation}
  \sum_\mu G_\mu^{nl} \geq \frac{1}{x} \sum_r e^{-r}\sum_{d(\mu,\nu)=r} G_\mu^{nl} +G_\nu^{nl} ,
\end{equation}
where $x = 2 \sum_r e^{-r} C_d r^d=O(1)$ is the number of times a single term is counted, and $d$ is the dimension of the lattice. Therefore, it is enough to prove that for a given $r$ and any pair $\mu,\nu$ with $d(\mu,\nu)=r$,
 \begin{equation}\label{eq:gap_GG}
 G_\mu^{nl} G_\nu^{nl} \geq - \Delta'\frac{1}{x} e^{-r} (G_\mu^{nl} + G_\nu^{nl}).
 \end{equation}
Note that the kernel of the LHS of \eqref{eq:gap_GG} is contained in the kernel of the RHS. Next, $G^{nl}_\mu+G^{nl}_\nu$ can be lower bounded by
\begin{equation}
  G^{nl}_\mu+G^{nl}_\nu \geq \frac{1}{2} \left(1-P_{Ker(G^{nl}_\mu+G^{nl}_\nu)}\right),
\end{equation}
since $G^{nl}_\mu+G^{nl}_\nu\approx P_\mu + P_\nu$, which has gap 1, and at sufficiently high (but constant) temperature the difference is sufficiently small.

To lower bound the l.h.s of \eqref{eq:gap_GG}, write $G_\mu^{nl} = \left|G_\mu^{\lceil r/2 \rceil}\right|+X_\mu$ and $G_\nu^{nl} = \left|G_\nu^{\lceil r/2 \rceil}\right|+X_\nu$. $\left|G_\nu^{\lceil r/2 \rceil}\right|$ and $\left|G_\mu^{\lceil r/2 \rceil}\right|$ commute as they are supported on two disjoint regions, and they are positive, thus their product is also positive. The norm of $X_\mu,X_\nu$ is bounded by (with $y=e^\eta\left(e^{4\beta}-1\right)$ as defined in the proof of Theorem~\ref{thm:clusterLogLocal})
\begin{equation}
	\|X_\mu\|= \|G_\mu -|G_\mu^{\lceil r/2 \rceil}|\|\leq \|G_\mu -G_\mu^{\lceil r/2 \rceil}\|+ \|G_\mu^{\lceil r/2 \rceil}-|G_\mu^{\lceil r/2 \rceil}|\| \leq 3 \frac{y^{\lceil r/2 \rceil}}{1-y},
\end{equation}
since $\|G_\nu^{nl} - G_\nu^{\lceil r/2 \rceil}\|\leq y^{\lceil r/2 \rceil}/{(1-y)}$ by \eqref{eq:clustertrunc_error}, and thus $G_\nu^{\lceil r/2 \rceil}\geq -y^{\lceil r/2 \rceil}/{(1-y)}$, so $\left|G_\nu^{\lceil r/2 \rceil}-|G_\nu^{\lceil r/2 \rceil}|\right|\leq 2y^{\lceil r/2 \rceil}/{(1-y)}$. 
Using that above some constant temperature $\|G_\mu^{r}\|<2$, we get that 
\begin{equation}
  \left\|G_\mu^{nl} G_\nu^{nl} -|G_\mu^{\lceil r/2 \rceil}||G_\nu^{\lceil r/2 \rceil}|\right\|\leq 18 \frac{y^{\lceil r/2 \rceil}}{1-y}\leq \Delta'\frac{1}{2x} e^{-r}
\end{equation}
for sufficiently small (but constant) $\beta$. 
 Therefore for any $\mu, \nu$ pair, the following is true:
 \begin{equation}
 G_\mu^{nl} G_\nu^{nl} \geq - 18 \frac{y^{\lceil r/2 \rceil}}{1-y} [1-P_{Ker (G_\mu^{nl} G_\nu^{nl})}]\geq - \Delta'\frac{1}{2x} e^{-r}  [1-P_{Ker (G_\mu^{nl} +G_\nu^{nl})}]\geq - \Delta'\frac{1}{x} e^{-r} (G_\mu^{nl} + G_\nu^{nl}),
 \end{equation}
 as the kernel of $G_\mu^{nl} G_\nu^{nl}$ contains the kernel of $G_\mu^{nl}+G_\nu^{nl}$. 
 This proves Eq.~\eqref{eq:gap_GG} and thus Theorem~\ref{thm:gap}. 
\end{proof}

\end{document}